\documentclass[conference]{IEEEtran}
\makeatletter
\def\endthebibliography{%
	\def\@noitemerr{\@latex@warning{Empty `thebibliography' environment}}%
	\endlist
}
\makeatother
\IEEEoverridecommandlockouts
\usepackage{amsmath,amsthm,amssymb,amsfonts}
\usepackage{algorithmic}
\usepackage{graphicx}
\usepackage{textcomp}
\usepackage{xcolor}
\usepackage[sort&compress,numbers]{natbib}
\usepackage{booktabs} % For formal tables
\usepackage{subfig}
\usepackage[ruled,linesnumbered, noline]{algorithm2e}
\usepackage[misc]{ifsym}
\usepackage{amsfonts}
\usepackage{graphicx}
\usepackage{lineno}
\usepackage{array}
\usepackage{longtable}
\usepackage{float}
\newtheorem{theorem}{Theorem}
\newtheorem{lemma}{Lemma}

\theoremstyle{definition}
\newtheorem{definition}{Definition}

\theoremstyle{remark}

\def\x{{\textbf x}}

\def\o{{\textbf o}}
\def\s{{\textbf s}}
\def\u{{\textbf u}}
\def\NkSMK{\mathsf{\mbox{non-monotone } kSMK}}
\def\kSMK{\textsf{kSMK}}
\def\kIMK{\textsf{kIMK}}

\def\kSPK{\textsf{kSPK}}

\def\LAA{\textsf{LAA}}
\def\RLA{\textsf{RLA}}

\def \y{{\textbf y}}
\def \0{{\textbf 0}}
\def \opt{{\mathsf{opt}}}
\def \eE{\mathbb{E}}
\begin{document}
	\title{Robust Approximation Algorithms for Non-monotone $k$-Submodular Maximization under a Knapsack Constraint}
	\author{\IEEEauthorblockN{1\textsuperscript{st} Dung T.K. Ha}
		\IEEEauthorblockA{\textit{Faculty of Information Technology}\\
			\textit{VNU University of Engineering and Technology}\\
			Hanoi, Vietnam\\
			20028008@vnu.edu.vn}
		\and
		\IEEEauthorblockN{2\textsuperscript{nd} Canh V. Pham (\Letter)}
		\IEEEauthorblockA{\textit{ORLab, Phenikaa University} \\
			Hanoi, Vietnam \\
			canh.phamvan@phenikaa-uni.du.vn}
		\and
		\IEEEauthorblockN{3\textsuperscript{rd} Tan D. Tran}
		\IEEEauthorblockA{\textit{Faculty of Information Technology}\\
			\textit{VNU University of Engineering and Technology}\\
			Hanoi, Vietnam\\			
			22027005@vnu.edu.vn}
		\and
		\IEEEauthorblockN{4\textsuperscript{th}Huan X. Hoang}
		\IEEEauthorblockA{\textit{Faculty of Information Technology}\\
			\textit{Halong University}\\ Quang Ninh, Vietnam\\
			and \textit{VNU University of Engineering and Technology}\\
			Ha Noi, Vietnam\\
			huanxhoang@vnu.edu.vn}
		%\and
	}
	\maketitle
	\begin{abstract}
		The problem of non-monotone $k$-submodular maximization under a knapsack constraint ($\kSMK$) over the ground set size $n$ has been raised in many applications in machine learning, such as data summarization, information propagation, etc.  However, existing algorithms for the problem are facing questioning of how to overcome the non-monotone case and how to fast return a good solution in case of the big size of data. This paper introduces two deterministic approximation algorithms for the problem that competitively improve the query complexity of existing algorithms.
		Our first algorithm, $\LAA$,  returns an approximation ratio of $1/19$ within $O(nk)$ query complexity. The second one, $\RLA$, improves the approximation ratio to $1/5-\epsilon$ in $O(nk)$ queries, where $\epsilon$ is an input parameter. 
		Our algorithms are the first ones that provide constant approximation ratios within only $O(nk)$ query complexity for the non-monotone objective. They, therefore, need fewer the number of queries than state-of-the-the-art ones by a factor of $\Omega(\log n)$.
		
		Besides the theoretical analysis, we have evaluated our proposed ones with several experiments in some instances:  Influence Maximization and Sensor Placement for the problem. The results confirm that our algorithms ensure theoretical quality as the cutting-edge techniques and significantly reduce the number of queries.
	\end{abstract}
	
	\begin{IEEEkeywords}
		Approximation algorithm,	$k$-submodular maximization, knapsack constraint, non-monotone.
	\end{IEEEkeywords}
	\section{Introduction}
	\quad
	\label{sec:intro}
	%The problems of constrained $k$-submodular function maximization have played an important role in combinatorial optimization and machine learning recently because of their natural usages in various domains 
	% Due to their widespread natural applications across numerous domains, constrained 
	$k$-submodular is a generalized version of submodular in polyhedra \cite{Lovasz82} in which some properties of submodularity have deep theoretical extensions to $k$-submodularity that challenge researchers to study \cite{bi-sub-12,ksub_uncon_soda14,ksub-globecom19}, etc. Maximizing a $k$-submodular function subject to some constraints has recently become crucial in combinatorial optimization and machine learning such as influence maximization via social networks \cite{ksub-nip15,ksub-icml20,ksub-tevc18,ksub-stream-icml20}, sensor placement \cite{ksub-nip15,ksub-icml20,ksub-tevc18}, feature selection \cite{bi-sub-12} and information coverage maximization \cite{ksub-tevc18}, etc. Given a finite ground set $V$ with $|V|=n$, and an integer number $k$, let $[k]=\{1, 2, \ldots, k\}$, and  $(k+1)^V=\{(V_1, V_2, \ldots, V_k)| V_i \subseteq V, \forall i \in [k], V_i \cap V_j =\emptyset, \forall i \neq j\}$ be a family of $k$ disjoint sets, called the \textbf{$k$-set}. We have the following definition of the $k$-submodular function:
	\begin{definition}[$k$-submodularity \cite{ksub_uncon_soda14}]
		A function $f: (k+1)^V \mapsto \mathbb{R}_+$ is $k$-\textbf{submodular }iff for any $\x=(X_1, X_2, \ldots, X_k)$ and $\y=(Y_1, Y_2, \ldots, Y_k)$ $\in  (k+1)^V$, we have:
		\begin{align}
			f(\x)+ f(\y) \geq f(\x \sqcap \y) + f(\x \sqcup \y )
		\end{align}
		where
		$$\x \sqcap \y=(X_1 \cap Y_1, \ldots, X_k \cap Y_k)$$ and 	$$\x \sqcup \y=(Z_1, \ldots, Z_k),\ \mbox{where} \ Z_i= X_i \cup Y_i \setminus (\bigcup_{j \neq i} X_j \cup Y_j)$$ 
	\end{definition}
	%For any $\x \in (k+1)^V$, $e\notin supp(\x)$ and $i \in [k]$, we have the \textit{marginal gain} when adding an element $e$ to the $i$-set $X_i$ of $\x$ is:
	%\begin{align*}
	%	\Delta_{(e, i)} f(\x)=&f(X_1, \ldots, X_{i-1}, X_i \cup \{e\}, X_{i+1}, \ldots, X_k) 
	%	\\
	%	&-f(X_1, \ldots, X_k)
	%\end{align*}
	In this paper, we consider the problem of $k$-Submodular Maximization under a Knapsack constraint ($\kSMK$) %which is one of the most natural and general $k$-submodularity optimizations since it captures the limitation of budget, time, or size when choosing elements.
	which is defined as follows:
	\begin{definition}[The $k$-Submodular Maximization under a Knapsack constraint ($\kSMK$) problem]
		Under the knapsack constraint, each element  $e$ is assigned a positive cost $c(e)$. Given a limited budget $B>0$, the problem $\kSMK$ asks to find a $k$-set $\x=(X_1, X_2, \ldots, X_k)$ with total cost $c(\x)=\sum_{e \in X_i, i\in [k]}c(e)\leq B$ so that $f(\x)$ is maximized. 
	\end{definition}
	The problem $\kSMK$ is a general model applied to a lot of essential instances such as $k$-topic influence maximization, $k$-type sensor placement, $k$-topic information coverage maximization \cite{Wang_kInfoCover,QianSTZ18,DSAA22}, etc., with the knapsacks that encode users’ constraints including budget, time or size. For example, \textbf{$k$-topic influence maximization under knapsack constraint ($\kSMK$)} \cite{ksub_uncon_soda14,QianSTZ18,ksub-icml20}, the problem asks for maximizing the expected number of users, who are influenced by at least one of $k$ distinct topics with a limited budget $B>0$. The mathematical nature is $k$-submodular maximization under a diffusion model, which Kempe {\em et al.}\cite{kem03} first proposed with a single type of influence. %Authors in \cite{ksub-nip15} generalized this model to allow $k\geq2$ types of influence, then other authors were also attracted by $k$-topic influence maximization \cite{ksub-icml20,ksub-tevc18,ksub-stream-icml20}. 	
	%%%% The needs of linear query complexity desterministc alg
	
	The challenge when providing a solution for $\kSMK$ is it has many candidate approximate solutions with different sizes. We have to select the best nearly optimal one within polynomial time. Therefore, beyond obtaining a nearly optimal solution to $\kSMK$ in the aforementioned applications, designing such a solution must also minimize the query complexity, especially for big data, since the tremendous amount of input data makes the search space for a solution crazily soar. Unfortunately, $k$-submodularity requires an algorithm to evaluate the objective function whenever observing an incoming element. Therefore, it is necessary to design efficient algorithms in reasonable computational time. We refer to the \textit{query complexity} as a  measure of computational time since it dominates the time running of an algorithm. Previous works \cite{TANG202228,ksub-knap,DSAA22} proposed efficient algorithms for $\kSMK$ in which even some algorithms can provide solutions in linear query complexity of $O(kn)$. However, these works are just available for the monotone case. Meanwhile, some works \cite{Feige_nonmono_submax,Buchbinder_Querytradeoff,fast_icml} showed that the $k$-submodular objective function might be non-monotone in practical applications. Therefore, solving the $\NkSMK$ problem within linear query complexity is critical.
	
	Overall, this paper aims to tackle both challenges above for non-monotone $k$-submodular maximization and constrained by a knapsack.   %For instance, data summarization is a kind of submodular optimization problem\cite{data-sum-1,data-sum-2,ap-datasum,fast_icml,TschiatschekIWB14}. It explores a representative subset for a whole set of data points. The representative subset satisfies two conditions which are coverage and diversity \cite{TschiatschekIWB14}. However, authors \cite{fast_icml} analyzed that high coverage tended to select more elements, whereas higher diversity pointed to eliminating too many similar elements appearing in the representative subset and preventing the summary from growing too large. Hence, the objective function designed to measure both coverage and diversity is naturally non-monotone. Consequently, the $k$-submodular objective function may be non-monotone. 
	\subsection{Our contribution}
	In this work, we design novel approximate algorithms that respond to some requirements about providing considerable solution quality and reducing query complexity. In particular, our work is the first one that provides a constant approximation ratio within only $O(kn)$ query complexity for $\NkSMK$. The main version, $\RLA$ returns an approximation ratio of $1/5-\epsilon$ which is equivalent to the state-of-the-art one proposed in  \cite{Canh_joco21}.
	In general, our contributions are as per the following:
	\begin{itemize} 
		\item We first propose the $\LAA$ algorithm (Algorithm~\ref{alg:1}), a $1/19$-approximation one that scans a single pass over the ground set within $O(kn)$ query complexity. It's the first simple but vital algorithm of our work since it limits the range of the optimal value. Besides, it provides a data division strategy to reduce query complexity to $O(nk)$. 
		\item We next propose $\RLA$ algorithm (Algorithm~\ref{alg:2}) that achieves an approximation ratio $1/5-\epsilon$, and requires $O(kn/\epsilon)$ query complexity where $\epsilon>0$ is an accuracy parameter. 
		Specifically, to the best of our knowledge, our algorithm is also equivalent to the current best approximation ratio of a deterministic algorithm for the studied problem in~\cite{Canh_joco21}.  		
		\item To illustrate the theoretical contributions, we conduct several comprehensive experiments in two applications of $\kSMK$ including $k$-topic Influence Maximization and $k$-type Sensor Placement. Experimental results have shown that our algorithms save queries more than state-of-the-art (mentioned in Table~\ref{tab:1}) and return comparable results in terms of performance.
	\end{itemize} 
	Table~\ref{tab:1} compares our algorithms with some state-of-the-art algorithms for $\NkSMK$ on three aspects, including approximation ratio, query complexity, and deterministic or not. These fields indicate that our algorithms have both a low number of queries and valuable deterministic approximation ratios that are equivalent to or even better than the others.
	\begin{table*}[hpt]
		\centering
		\footnotesize
		\begin{tabular}{llll}
			\hline
			\textbf{Reference} &  \textbf{Approximation ratio} & \textbf{Query complexity}&\textbf{Is deterministic?}
			\\
			\hline
			\textbf{$\LAA$ (Alg.~\ref{alg:1}, this paper)}&  $1/19$ & $O(kn)$ & Yes
			\\
			\textbf{$\RLA$ (Alg.~\ref{alg:2}, this paper)} &  $1/5-\epsilon$ & \textbf{$O(kn/\epsilon)$} & Yes
			\\
			%\textbf{$\RLAP$ (Alg.~\ref{alg:3}, this paper)} & $1/3-\epsilon$ &$O(kn \log(1/\epsilon)/\epsilon)$ & Deterministic   
			%\\	
			Deterministic Streaming\cite{Canh_joco21}    & $1/5-\epsilon$ & \textbf{$O(kn\log(n)/\epsilon)$} & Yes
			\\	
			Random Streaming\cite{Canh_joco21}    & $k/(5k-2)-\epsilon$ & $O(kn\log(n)/\epsilon)$ &  No
			%\\
			%Greedy\cite{TANG202228}   & $1/2-1/(2e)$ & $O(n^4k^3)$ & Yes
			%\\
			%Algorithm 3 \cite{ksub-knap} & $1/2$ & $poly(n)$ & Random & Yes
			\\
			\hline
		\end{tabular}
		\caption{Algorithm comparison for $\NkSMK$; Note that Desterministic Streaming and Random Streaming in \cite{Canh_joco21} is the special case when $\beta=1$.}
		\label{tab:1}
	\end{table*}
	
	\textit{Organization} The rest of the paper is organized as follows: We provide a literature review and discussions in Section~\ref{sec:related}.
	The notations and properties of $k$-submodular functions are presented in Section~\ref{sec:pre}. Section~\ref{sec:alg} presents our algorithms and theoretical analysis. The extensive experiments are shown in Section~\ref{sec:ex}. Finally, we conclude this work in Section~\ref{sec:con}.
	\section{Related work}
	\label{sec:related}
	In this section, we review related works and provide some discussion on existing algorithms. 
	
	Studying $k$-submodular functions appears when considering the submodularity in polyhedra. Lov{\'{a}}sz \cite{Lovasz82} found that it was a similar but deeper theory than submodularity when working with intersection matroids. %After that, people considered a bisubmodular function to answer the above question that created the foundation for $k$-submodularity research. Singh {\em et al.}\cite{bi-sub-12} worked with bisubmodular maximization, which meant $k$-submodularity with $k=2$. 
	After that, more works have focused on the issue of $k$-submodularity with general $k\ge 2$. %Certainly, with $k=1$, the issue becomes submodular maximization. As submodular maximization is NP-hard, $k$-submodular maximization is also NP-hard.	
	%	People studied $k$-submodularity with both unconstrained maximization \cite{ksub_uncon_soda14,ksub_uncon_soda16} and maximization under constraints such as cardinality \cite{ksub-icml20,ksub-nip15,ksub-tevc18,ksub-stream-icml20}, knapsack \cite{Canh_joco21,TANG202228}, and matroid \cite{ksub-jdo16,ksub_matroid-icml20}.
	First, people studied maximizing unconstrained $k$-submodular~\cite{ksub_uncon_soda14, ksub_uncon_soda16,ksub_iwoca-Oshima17}. %studied to . They devised a deterministic Greedy algorithm with an approximation ratio of $1/3$. Next, the authors in~\cite{ksub_uncon_soda16} presented a random Greedy approach based on a probability distribution which improved the approximation ratio to $k/(2k-1)$. %select which element has a larger marginal gain with higher probability. Oshima {\em  et al.}\cite{ksub_iwoca-Oshima17} %eliminated the random told in \cite{ksub_uncon_soda16}proposed a derandomization method, however, the number of queries expanded to $O(n^2k^2)$. 
	%The constrained maximizing of the $k$-submodular function has been researched further. 
	Due to the practical values when solving the problem with constraints, some authors focused on $k$-submodular maximization under some kinds of constraints \cite{ksub-nip15,ksub-jdo16, ksub_matroid-icml20}, etc. Authors focused on the monotone case such as Oshaka {\em et al.}~\cite{ksub-nip15} studied monotone $k$-submodular maximization with two kinds of size constraint: overall size constraint and singular size constraint, authors in~\cite{QianSTZ18} proposed a multi-objective evolutionary method to provide an approximation ratio of $1/2$ for the monotone $k$-submodular maximization problem with the overall size constraint. However, this algorithm took a high query complexity of $O(kn\log^2B)$ in expectation. Authors~\cite{ksub-Soma19} further proposed an online algorithm with the same approximation ratio of $1/2$ but runs in polynomial time with regret bound. However, these contributions just work for the monotone case and for size constraints; hence, it's hard to adapt to $\NkSMK$. Moreover, these algorithms required exponential running time~\cite{ksub-nip15} or high query complexity~\cite{QianSTZ18}. 
	
	Recently, Nguyen {\em et al.}\cite{ksub-stream-icml20} first applied streaming to solve the problem of $k$-submodular maximization with overall size constraint. Streaming fashion is an active approach when it requires only a small amount of memory to store data and scans one or a few times over the ground set $V$. They devised two streaming algorithms within $O(nk\log(k))$ query complexity. Their first one is deterministic and returns an approximation ratio of $1/3-\epsilon$, while the second one is randomized and returns an approximation ratio of $k/(3k-1)-\epsilon$. Later on, Ene and Nguyen~\cite{Ene22} developed a single-pass streaming algorithm based on integer programming formulation for $k$-submodular maximization with singular size constraint with an approximation ratio of $0.5/(1+B(2^{1/B}-1))$ within $O(nk)$ queries, where $B=\min_{i \in [k]}B_i$. 
	
	Unlike cardinality or matroid, which just enumerates elements, the knapsack requires maximizing $f(\cdot)$ subject to a given budget that the total cost of a solution can not exceed. %The Knapsack constraints do not allow for just enumerating elements like cardinality or matroid constraints. 
	Hence, there can be multiple maximal cost solutions that are not the same size. The authors \cite{ksub-knap} proposed a multi-linear extension method with an approximation ratio of $1/2-2\epsilon$ in expectation for the $\kSMK$. %They employed a 2-step method and constructed a continuous extension of the discrete problem. A rounding technique extracts a discrete solution from a fractional one after an optimization algorithm located an optimum in the continuous space. 
	This work provides the best approximation ratio in expectation. However, this algorithm is impractical due to the high query complexity of a continuous extension~\cite{BalkanskiQS21}.
	
	Besides, Wang {\em et al.}~\cite{TANG202228} proposed a $(1/2-1/(2e))$-approximation algorithm for the $\kSMK$ that inspired from the Greedy algorithm in~\cite{Sviridenko04}. This algorithm, however, requires an expensive query complexity of $O(n^4k^3)$, and therefore it is difficult to apply to medium-sized instances even though one can compute the objective function $f$ in $O(1)$ time. The authors~\cite{ksub-knap} proposed a multi-linear extension that provided the approximation ratio of $1/2-2\epsilon$ in expectation for the $\kSMK$. This work provides the best approximation ratio in expectation, however, it is impractical because of the high query complexity of a continuous extension~\cite{BalkanskiQS21}. Authors~\cite{DSAA22} first proposed a $(1/4-\epsilon)$-desterministic approximation algorithm within $O(kn/\epsilon)$. Nonetheless, the aforementioned works are not available for the non-monotone case. 
	
	To state the non-monotone $k$-submodularity, Pham {\em et al.}~\cite{Canh_joco21} recently have proposed two single-pass streaming algorithms for the $k$-submodular maximization under the budget constraint, a general of knapsack constraint within $O(nk\log(n)/\epsilon)$ queries. These algorithms returned the ratios of $1/5-\epsilon$ and $k/(5k-2)-\epsilon$ (in expectation) for the non-monotone case. Our best algorithm version, $\RLA$, gives an equivalent performance of them ($1/5-\epsilon)$ approximation ratio) yet reduces the query complexity to $O(kn/\epsilon)$. 
	
	On the whole, the characteristic of our algorithms is deterministic, linear query complexity, and available for non-monotonicity. 
	
	\section{Preliminaries}
	\label{sec:pre}
	\textbf{Notations.} 
	Given a ground set $V=\{e_1, e_2, \ldots, e_n \}$ and an integer $k$, we define $[k]=\{1, 2, \ldots, k\}$ and let $(k+1)^V=\{(V_1, V_2, \ldots, V_n)|V_i \subseteq V  \ \forall i \in [k], V_i\cap V_j=\emptyset \ \forall i \neq j\}$ be a family of $k$ disjoint subsets of $V$, called $k$-set.
	
	For $\x=(X_1, X_2, \ldots, X_k)\in (k+1)^V$, we define $supp_i(\x)=X_i$,  $supp(\x)=\cup_{i\in [k]}X_i$, $X_i$ as \textbf{$i$-th set of $\x$} and an empty $k$-set $\0=(\emptyset, \ldots, \emptyset)$.
	We set if $e \in X_i$ then $\x(e)=i$ and $i$ is called the \textbf{position} of $e$ in $\x$, otherwise $\x(e)=0$. Adding an element $e \notin supp(\x)$ into $X_i$ can be represented by $\x  \sqcup  (e, i) $. We also write $\x=\{(e_1, i_1), (e_2, i_2), \ldots, (e_t, i_t)\}$ for $e_j \in supp(\x), i_j=\x(e_j), \forall 1 \leq j\leq t$.
	When $X_i=\{e\}$, and $X_j= \emptyset, \forall j\neq i$, $\x$ is denoted by $(e,i)$.
	
	For 	$\x=(X_1, X_2, \ldots, X_k), \y=(Y_1, Y_2, \ldots, Y_k) \in (k+1)^V$, we denote by $\x \sqsubseteq \y$ iff $X_i \subseteq Y_i$ $\forall i\in [k]$. 
	
	\textbf{The objective function.} 
	The function $f: (k+1)^V \mapsto \mathbb{R}_+$ is $k$-\textbf{submodular} iff for any $\x=(X_1, X_2, \ldots, X_k)$ and $\y=(Y_1, Y_2, \ldots, Y_k)$ $\in  (k+1)^V$, we have:
	\begin{align}
		f(\x)+ f(\y) \geq f(\x \sqcap \y) + f(\x \sqcup \y )
	\end{align}
	where
	$$\x \sqcap \y=(X_1 \cap Y_1, \ldots, X_k \cap Y_k)$$ and 
	$$\x \sqcup \y=(Z_1, \ldots, Z_k),\ \mbox{where} \ Z_i= X_i \cup Y_i \setminus (\bigcup_{j \neq i} X_j \cup Y_j)$$ 
	For any $\x \in (k+1)^V$, $e\notin supp(\x)$ and $i \in [k]$, we have the \textit{marginal gain} when adding an element $e$ to the $i$-set $X_i$ of $\x$ is:
	\begin{align*}
		\Delta_{(e, i)} f(\x)=&f(X_1, \ldots, X_{i-1}, X_i \cup \{e\}, X_{i+1}, \ldots, X_k) 
		\\
		&-f(X_1, \ldots, X_k)
	\end{align*}
	In this work, we consider $f$ to be \textit{non-monotone}, i.e., the marginal gain when adding a tuple $(e,i)$ to set $\x$, $\Delta_{(e, i)} f(\x)$, may be negative. we also assume that $f$ is normalized, i.e, $f(\0)=0$, and there exists an \textit{oracle query}, which when queried with the $k$-set $\x$ returns the value $f(\x)$. We also recap some properties of the $k$-submodular function that will be used for designing our algorithms.
	
	From~\cite{ksub_uncon_soda14}, a $k$-submodular function $f: (k+1)^V \mapsto \mathbb{R}_+$ is $k$-submodular iff it is pairwise monotone and orthant submodular. The $k$-submodularity of $f$ implies the \textit{orthant submodularity}, i.e., 
	\begin{align}
		\Delta_{(e,i)}f(\x) \geq \Delta_{(e,i) } f(\y)
	\end{align}
	for any $\x, \y \in (k+1)^V$ with $\x \sqsubseteq \y$, $e \notin supp(\y)$ and $i \in [k]$, and the \textit{pairwise monotonicity}, i.e.,
	
	\begin{align}
		\Delta_{(e,i)}f(\x) + \Delta_{(e,j)}f(\x) \geq 0 
	\end{align}
	for any $\x \in (k+1)^V$ with $e \notin supp(\x)$ and $i, j \in [k]$ with $i \neq j$.

	%In this paper, we assume that $f$ is normalized, i.e, $f(\0)=0$ and each element $e$ has a positive cost $c_i(e)$ to be added into $i$-th set of a solution and the total cost of $k$-set $\x$ is 
	%$$c(\x)=\sum_{i \in [k], e \in supp_i(\x)} c_i(e)$$
	%We define $\beta$ as the \textbf{largest gap between cost values of an element}, i.e, $$\beta=\max_{e\in V, i \neq j}\frac{c_i(e)}{c_j(e)}$$ 
	%Without loss of generality, throughout this paper, we assume that every element $e$ satisfies $c_i(e)\geq 1, \forall i \in [k]$ and $c_i(e)\leq B$ as otherwise we can simply remove it. 
	
	\textbf{The problem definition.} Assuming that each element $e$ is assigned a positive cost $c(e)$ and the total cost of a $k$-set $\x$ $c(\x)=\sum_{e \in supp(\x)}c(e)$. Given a limited budget $B>0$,  we assume that every item $e\in V$ satisfies $c(e) \leq  B$; otherwise, we can simply discard it. The $k$-Submodular Maximization under Knapsack constraint ($\kSMK$)  problem is to determine: 
	\begin{align}
		\arg \max_{\x \in (k+1)^V: c(\x)\leq B}f(\x).
	\end{align}
	It means the problem finds the solution $\x$ so that the total cost of $\x$ is less than or equal to $B$ so that $f(\x)$ is maximized. In this work,	we only consider $k\geq 2$ because if $k = 1$, the $k$-submodular function becomes the submodular function.
	
	%\textbf{$\kSMK$ problem.} The problem is defined as follows:
	%	\begin{definition}[$\kSMK$ problem]
		%	Given  a finite set $V$, a budget $B$ and a $k$-submodular function $f: (k+1)^V \mapsto \mathbb{R}_+$. The problem asks to find a solution $\s=(S_1, S_2, \ldots, S_k)$ in which an element $e \in V$ has a cost $c(e)>0$ when added into $S_i$, with total cost $c(\s)=\sum_{i \in [k], e \in supp_i(\s)}(e)\leq B$ so that $f(\s)$ is maximized.
		%	\end{definition}We denote by $\o=\{(o_1, i^*_1), \ldots, (o_m, i^*_m)\}$ an optimal solution of the problem, the optimal value $\opt=f(\o)$ and  $m=supp(\o)$. Without loss of generality, we assume that  $c(o_1)\geq c(o_2)\geq \ldots \geq c(o_m)$.
	\section{The algorithms}
	\label{sec:alg}
	In this section, we introduce two deterministic algorithms for $\kSMK$. The first algorithm, named \textbf{Linear Approximation Algorithm} ($\LAA$), has an approximation ratio of $1/19$  and takes $O(nk)$ query complexity. Although this approximation ratio is small, it is the \textbf{first one} that gives a constant approximation ratio within only $O(kn)$ queries for the non-monotone case. The approximation ratio is improved by our second algorithm, named \textbf{Robust Linear Approximation} ($\RLA$), from $1/19$ to $1/5-\epsilon$ by recalling the first algorithm's solution to provide a suitable range for bounding the optimal value $\opt$. Additionally, it scans the ground set $O(1/\epsilon)$ times and integrates the decreasing threshold strategy to get the near-optimal solution. 
	\subsection{Linear Approximation Algorithm}
	Our $\LAA$ algorithm adapts the idea of the recent work \cite{DSAA22} that (1) divides the ground set $V$ into two subsets: The elements with costs greater than $B/2$ are included in the \textbf{first subset},  while the remaining is included in the \textbf{second}, and (2) near-optimal solutions are sought and combined for the two aforementioned subsets. 
	\begin{algorithm}[h]
		\SetNlSty{text}{}{:}
		%	\begin{algorithmic}[1]
			\KwIn{$V$, $f$, $k$, $B>0$.}
			\KwOut{A solution $\s$}
			$\x \leftarrow \0$; $(e_{m}, i_{m}) \leftarrow (\emptyset, 1)$;            $\x'\leftarrow \0$; %\y\leftarrow 0$;
			\\
			\ForEach{$e \in V$}
			{
				$i_e \leftarrow \arg \max_{i \in [k]} f((e, i))$ \label{cond1}
				\\
				$(e_{m}, i_{m}) \leftarrow \arg\max_{(e', i')\in \{(e_m, i_m), (e, i_e)\}}f((e', i'))$
				\\
				\If{$c(e)\leq B/2$}
				{
					\If{$\Delta_{(e,i_e)}f(\x) \ge c(e)f(\x)/B$} 
					{ \label{cond2}
						$\x \leftarrow \x \sqcup  (e, i_e)$
					}
				}
			}
			$\x' \leftarrow \arg \max_{\x_j: j\leq t_x, c(\x_j)\leq B }c(\x_j)$, where $t_x=|supp(\x)|$  and  $\x_j=\{(e_{t_x-j+1}, i_{t_x-j+1}),(e_{t_x-j+2}, i_{t_x-j+2}), \ldots, (e_{t_x}, i_{t_x})\}$ is the last $j$ tuples added into $\x$.\\
			$\s \leftarrow \arg\max_{\s \in \{(e_{m} , i_{m}), \x'\}}f(\s)$
			\\
			\Return $\s_{final} $
			\caption{An Linear Approximation Algorithm ($\LAA$)}
			\label{alg:1}
		\end{algorithm}
		
		In particular, the algorithm first receives an instance $(V, f, k, B)$ of $\kSMK$ and initiates a candidate solution $\x$ as an empty set and a tuple $(e_m, i_m)$ as $(\emptyset, 1)$. The target of the tuple $(e_m, i_m)$ is to update the optimal solution found in the first subset, while the candidate solution $\x$ is to locate what solution is close to the optimal in the second. For each incoming element $e$, the algorithm finds ``the best" position $i_e$ in terms of the set $i$ in $k$ sets that returns the highest value $f((e, i_e))$. If its cost is greater than $B/2$, the role of $(e_m, i_m)$ is the best solution on the current first subset (line 5). Otherwise, the algorithm adds the tuple $(e,i_e)$ into  $\x$ if the condition $\Delta_{(e,i_e)}f(\x) \geq c(e)f(\x)/B$ is maintained. After the main loop completes, the algorithm selects a $k$-set $\x'$ as the set of last $j$ tuples adding into $\x$ with the maximum total cost nearest to $B$ (line 11). Finally, the algorithm returns the final solution $\s$ as the best one between $(e_m, i_m)$ and  $\x'$.  The details of the algorithm are fully presented in Algorithm~\ref{alg:1}.
		
	%	At a high level, the spirit of our algorithm is to resemble the ``divide and conquer" strategy in which it uses an appropriate subset division based on the costs of elements. The division of the ground set is productive because, during the linear time, the algorithm both finds the optimal solution on the first subset since feasible solutions have at most one element and finds the approximate solution on the second one. For the second subset, we were inspired by the suggestion from the idea of Kuhnle  {\em et al. }\cite{kuhnle_quickksubmono} in which elements with marginal gains that are over the \textit{ratio between the objective value of the current solution and the limited cost}, will be kept, and the remaining is released. Their idea is powerful in diminishing the number of queries of a constant factor approximation algorithm. However, to deal with the cost and the $k$-submodular function, we need to make a non-trivial analysis to give an approximation ratio.
		
		To deal with the non-monotonicity of the objective function, we have to use non-trivial analyzes to give an approximation. Differing from the monotone case in \cite{DSAA22}, we use the property of pairwise monotonicity as a critical component in our theoretical analysis. In the following, we analyze the theoretical guarantee of the Algorithm~\ref{alg:1}.
		We first define the notations as follows:  
		\begin{itemize}
			\item[$\bullet$] $V_1=\{e\in V: c(e) >B/2\}, V_2=\{e\in V: c(e) \leq B/2\}$. 
			\item[$\bullet$] $\o$ is an optimal solution of the problem over $V$ and the optimal value $\opt =f(\o)$.
			\item[$\bullet$] $\o'_1= \{(e, \o(e)) : e \in V_1 \}, \o'_2= \{(e, \o(e)): e \in V_2 \}$.
			\item[$\bullet$] $\o_1$ is an optimal solution of the problem over $V_1$.
			\item[$\bullet$] $\o_2$ is an optimal solution of the problem over $V_2$.
			\item[$\bullet$] $(e_j, i_j)$ as the $j$-th element added of the main loop of the Algorithm \ref{alg:1}.
			\item[$\bullet$] $\x=\{(e_1, i_1), \ldots, (e_t, i_t)\}$ the $k$-set $\x$ after ending the main loop, $t=|supp(\x|$.
			\item[$\bullet$]  $\x^j=\{(e_1, i_1),\ldots, (e_j, i_j)\}$: the $k$-set $\x$ (in the main loop) after adding $j$ elements $1\leq j\leq t$, $\x^0=\0$, $\x^t=\x$.
			\item[$\bullet$]  $\x_j=\{(e_{t-j+1}, i_{t-j+1}),(e_{t-j+2}, i_{t-j+2}), \ldots, (e_t, i_t)\}$ is the set of last $j$ elements added into $\x$.
			\item[$\bullet$] $\o_2^j=(\o_2 \sqcup \x^j ) \sqcup \x^j$.
			\item[$\bullet$] $\o_2^{j-1/2}=(\o_2 \sqcup \x^j ) \sqcup \x^{j-1}$. 
			\item[$\bullet$] $\x^{j-1/2}$: If $e_j \in supp(\o_2)$, then $\x^{j-1/2}=\x^{j-1} \sqcup (e_j, \o_2(e_j)) $. If $e_j \notin supp(\o_2)$, $\x^{j-1/2}=\x^{j-1}$.
			\item[$\bullet$] $\u^t=\{(u_1, i_1), (u_2, i_2), \ldots, (u_r,i_r) \}$ is a set of elements that are in $\o_2^t$ but not in $\x^t$, $r=|supp(\u^t)|$.
			\item[$\bullet$] $\u^t_l=\x^t \sqcup \{(u_1, i_1), (u_2, i_2), \ldots, (u_l,i_l) \},  1 \leq l\leq r$ and $\u^t_0=\x^t$.
		\end{itemize}
		%	Let $\s'$ be the remaining in $\s$ after pushing last elements from $\s$ into $\s$ such that the total cost of $\s$ does not exceed $B$. It means $\s=\s' \sqcup \s$. 
		Supposing that $\x'$ gets $T$ last tuples in $\x$, i.e.,  $\x'=\x_T$. Denote $Q=t-T$, we have $\x=\x^Q\sqcup\x'$. The following Lemmas connect the candidate solution $\x$ with $\o_2$.
		\begin{lemma}
			$f(\o_2)-f(\o_2^j) \leq 2f(\x^j)$ for all $0\leq j\leq t$.
			\label{lem:1}
		\end{lemma}
		\begin{proof} See the Appendix, section~\ref{sec:proof}
		\end{proof}
		\begin{lemma}
			$f(\x')\geq f(\x^t)/3$.
			\label{lem:2}
		\end{lemma}
		\begin{proof}  See the Appendix, section~\ref{sec:proof}
		\end{proof}
		\begin{lemma}
			$f(\o_2^t) \leq  4f(\x^t)$.
			\label{lem:3}
		\end{lemma}
		\begin{proof} See the Appendix, section~\ref{sec:proof}
	\end{proof}
	From these above lemmas, we imply the following lemma:
	\begin{lemma}
		$f(\x') \geq  f(\o_2)/18$.
		\label{lem:4}
	\end{lemma}
	\begin{proof}
		See the Appendix, section~\ref{sec:proof}
	\end{proof}
	\begin{theorem} Algorithm~\ref{alg:1} is a single-pass streaming algorithm that returns an approximation ratio of $1/19$ and takes  $nk$ queries.
		\label{theo:alg1}
	\end{theorem}
	\begin{proof} See the Appendix, section~\ref{sec:proof}
	\end{proof}
	\subsection{A Robust Linear Approximation Algorithm}
	We next introduce the $\RLA$ algorithm, which improves the approximation ratio to $1/5-\epsilon$ and takes $O(kn/\epsilon)$ query complexity. $\RLA$ keeps the key idea of $\LAA$ by reusing the $\LAA$'s solution to bounding the $\opt$'s range and adapts a greedy threshold to improve the approximation ratio by conducting $O(1/\epsilon)$ times scanning over the ground set. 	
	The details of the algorithm are fully presented in Algorithm~\ref{alg:2}.
	
	\begin{algorithm}[h]
		\SetNlSty{text}{}{:}
		%	\begin{algorithmic}[1]
			\KwIn{$V$, $f$, $k$, $B>0$, $\epsilon>0$.}
			\KwOut{A solution $\s$}
			$\s_b \leftarrow $ result of Algorithm~\ref{alg:1};
			$\Gamma \leftarrow f(\s_b)$
			%	$\theta \leftarrow 5 \Gamma/B$
			\\
			$A\leftarrow \{(1+\epsilon)^i: i \in \mathbb{N}, \Gamma \leq (1+\epsilon)^i\leq 19\Gamma \}$
			\\
			\For{$e \in V$}
			{
				\ForEach{$v \in A$}
				{
					$i_v \leftarrow  \arg\max_{i \in [k]}\Delta_{(e,i)}f(\s_v)$ \label{alg2_findmax}
					\\
					$\tau_v =2v/(5B)$				
					\\
					\If{$c(\s_v)+c(e) \leq B \  \mbox{\textbf{and}} \   \Delta_{(e, i_v)}f(\s_v)/c(e)) \geq \tau_v$} 
					{ 
						\label{alg2_condition}
						$ \s_v \leftarrow \s_v \sqcup  (e, i_v)$\\		
					}
				}			
			}
			$\s_{final} \leftarrow \arg\max_{\s' \in \{ \s_{max}, \s_1, \s_2, \ldots, \s_{|S|} \} } f(\s')$
			\Return $\s_{final}$
			\caption{Robust Linear Approximation ($\RLA$) Algorithm}
			\label{alg:2}
		\end{algorithm}
		
		Specifically,	$\RLA$ takes an instance $(V, f, k, B)$ of $\kSMK$ and an accuracy parameter $\epsilon>0$ as inputs.  $\RLA$ first calls  $\LAA$ as a subroutine and uses  $\LAA$'s solution, $\s_b$,  to obtain a bound range of the optimal solution (line 1). From Theorem \ref{theo:alg1}, we have $\Gamma \leq \opt \leq 19 \Gamma$.
		
		The major part of the algorithm consists of two loops: the outer to scan each element $e$ in the ground set $V$ and the inner to consider each candidate solution $\s_v$ for each $v$ filtered out from the set $A$. On the basis of Theorem \ref{theo:alg1}, we construct the set $A$ to bound the number of candidate solutions $\s_v$. We define $(e,i_v)$ as the tuple that gives the largest marginal gain when added into $\s_v$. When an element $e$ arrives, the algorithm handles these works: (1) choose the position $i_v$ with maximal marginal gain with respect to $\s_v$ and $e$ (line~\eqref{alg2_findmax}); (2) use threshold $\tau_v=2v/(5B)$ to add the element  $e$ into  $\s_v$ if it has the high \textit{density gain} that is defined as the ratio of the marginal gain of that element over its cost without violating the budget constraint (line~\eqref{alg2_condition}).
		
	We still keep the notations $\o$ as an optimal solution of the problem over $V$ and the optimal value $\opt =f(\o)$. We add some notations regarding to Algorithm~\ref{alg:2} as follows:
		\begin{itemize}
			\item[$\bullet$] $\s_v=\{(e_1, i_1), (e_2, i_2), \ldots, (e_q, i_q)\}$ is the candidate solution  with respect to some elements $v\in A$ after ending the outer loop.
			\item[$\bullet$] $\s^j_v=\{(e_1, i_1), (e_2, i_2), \ldots, (e_j, i_j)\}, 1 \leq j\leq q$ and $\s^0_v=\0$. 
			\item[$\bullet$]  $\s_v^{<e}$ as $\s_v$ right before $e$ is processed.
			\item[$\bullet$] 	$\u=\{(u_1, i_1), (u_2, i_2), \ldots, (u_r,i_r) \}$ as a set of elements belongs to $\o$ yet doesn't belong to $\s_v$, $r=|supp(\u)|$. 
			\item[$\bullet$] $\u_l=\s_{v} \sqcup \{(u_1, i_1), (u_2, i_2), \ldots, (u_l,i_l) \}, \forall 1 \leq l\leq r$ and $\u_0=\s_v$.
			\item[$\bullet$] $\o^j=(\o \sqcup \s^j ) \sqcup \s^j$.
			\item[$\bullet$] $\o^{j-1/2}=(\o \sqcup \s^j ) \sqcup \s^{j-1}$. 
			\item[$\bullet$] $\s^{j-1/2}$: If $e_j \in supp(\o)$, then $\s^{j-1/2}=\s^{j-1} \sqcup (e_j, \o(e_j)) $. If $e_j \notin supp(\o)$, $\s^{j-1/2}=\s^{j-1}$.
		\end{itemize}
		\begin{lemma} For any $v\in A$,	if there is no element $o \in supp(\o)\setminus supp(\s_v)$ so that $\Delta_{(o,\o(o))}f(\s^{<o}_v)\geq \tau_v$ and $c(\s_v^{<o})+c(o)>B$, we have: $f(\o)\leq 3f(\s_v)+c(\o)\tau_v$. 
			\label{lem:alg2.1}
		\end{lemma}	
		\begin{proof} See the Appendix, section~\ref{sec:proof}
		\end{proof}
		\begin{theorem}For  $ 0<\epsilon < 1/5$, the Algorithm~\ref{alg:2}  returns an approximation ratio of $1/5-\epsilon$, within   $O(nk/\epsilon)$ queries.
		\label{theo:alg2}
		\end{theorem}
		\begin{proof}
		 See the Appendix, section~\ref{sec:proof}
		\end{proof}
		\section{Experiments}
		\label{sec:ex}
		In this section, we compare the performance between our algorithms and state-of-the-art algorithms for the $\kSMK$ problem listed below:
		\begin{itemize}
			%\item \textbf{Greedy}: the $(1/2-1/(2e))$-approximation algorithm within $O(n^4k^3)$ in \cite{TANG202228}.
			\item \textbf{Deterministic Streaming (DS)}\footnote{The $\kSMK$ problem is a special case of the $k$-submodular maximization under the budget constraint in \cite{Canh_joco21} with $\beta =1$.}: A streaming algorithm in \cite{Canh_joco21} which returns an approximation ratio of $1/5-\epsilon$, requires 1-pass and $O(kn\log(n)/\epsilon)$ queries.
			\item \textbf{Random Streaming (RS)}: Another streaming algorithm in \cite{Canh_joco21} which returns an approximation ratio of $k/(5k-2)-\epsilon$ in expectation, requires one pass and $O(kn\log(n)/\epsilon)$ queries.
		\end{itemize}
		Although Greedy proposed by \cite{TANG202228} gives the best approximation ratio yet it is only available for the monotone case. Besides, in \cite{DSAA22}, authors also showed that the running time of Greedy was so long that they had to limit the time to cut off the experiment. Therefore, in the experiment, we eliminated the Greedy. Also, we conduct experiments on specific applications, which are \textbf{$k$-topic Influence Maximization under knapsack constraint ($\kIMK$)} and \textbf{$k$-type Sensor Placement under Knapsack constraint ($\kSPK$)} on three important measurements: the oracle value of the objective function, the number of queries, and running time. We further show the trade-off between the solution quality and the number of queries of algorithms with various settings of budget $B$. 
		
		We also use the dataset as mentioned in~\cite{ksub-stream-icml20} to illustrate the performance of compared algorithms (Table~\ref{tab_db}). To demonstrate the performance of algorithms via the above three measurements, we show some figures numbered and captioned, in which the terms Fig, K, and M stand for Figure, thousands, and millions, respectively.  
		
		All the implementations are on a Linux machine with configurations of $2\times$ Intel Xeon Silver $4216$ Processor @$2.10$GHz and $16$ threads x$256$GB DIMM ECC DDR4 @$2666$MHz. 
		\begin{table}[hpt]
			\caption{The dataset}
			\label{tab_db}   
			\centering
			\begin{tabular}{lccccc}
				\hline
				\textbf{Database} & \textbf{\#Nodes} & \textbf{\#Edges} &\textbf{ Types} &\textbf{ Instances}
				\\ 
				\hline
				Facebook \cite{fb_network} & 4039 & 88234 & directed& $\kIMK$ \\
				Intel Lab sensors\cite{sensor} & 56 & - & -&\kSPK \\
				\hline
			\end{tabular}
		\end{table}
		\subsection{$k$-topic Influence Maximization under Knapsack constraint ($\kIMK$)}
		\quad The information diffusion model, called Linear Threshold (LT) model \cite{kem03,ksub-stream-icml20} was briefed, and the $k$-topic Influence Maximization under Knapsack constraint ($\kIMK$) problem using this model was defined as follows:
		\paragraph{LT model} A social network is modeled by a directed graph $G=(V, E)$, where $V, E$ represent sets of users and links, respectively. Each edge $(u, v)\in E$ is assigned weights $\{ w^i(u, v) \}_{ i\in [k]}$, where each $w^i(u, v)$ represents how powerful $u$ influences to $v$ on the $i$-th topic.  Each node $u \in V$ has a \textit{influence threshold} with topic $i$, denoted by $\theta^i(u)$,  which is chosen uniformly at random in $[0,1]$. Given a seed set $\s=(S_1, S_2, \ldots, S_k) \in (k+1)^V$, the information propagation  for topic $i$ happens in discrete steps $t=0, 1,\ldots$ as follows. At step $t=0$, all nodes in $S_i$ become active by topic $i$. At step $t\geq 1$, a node $u$ becomes active if $\sum_{\mbox{actived} v} w^i(v,u) \geq \theta^i(u)$. 
		
		The information diffusion process on topic $i$ ends at step $t$ if there is no new active node and the diffusion process of a topic is independent of the others. Denote by $\sigma(\s)$ the number of nodes that become active in at least one of $k$ topics after the diffusion process of a seed $k$-set $\s$, i.e.,
		\begin{align}
			\sigma(\s)=\eE[|\cup_{i\in [k]} \sigma_i(S_i)|]
			\label{IM}
		\end{align}
		where $\sigma_i(S_i)$ is a random variable representing the set of active users for topic $i$ with the seed $S_i$. 
		\paragraph{The $\kIMK$ problem} The problem is formally defined as follows:
		\begin{definition}[$\kIMK$ problem]  Assuming that each user $e$ has a cost $c(e)>0$ for every $i$-th topic, which illustrates how difficult it is to initially influence the appropriate individual about that topic.  Given a budget $B>0$, the problem asks to find a seed set $\s$ with $c(\s)=\sum_{e \in S_i, i\in k} c(e)\leq B$ so that $\sigma(\s)$ maximal.
		\end{definition}  
		\subsection{$k$-type Sensor Placement under Knapsack constraint}
		\quad We further study the performance of algorithms for \textit{$k$-type Sensor Placement under Knapsack constraint} ($\kSPK$) problem which is formally defined as follows:
		\begin{definition}[$\kSPK$ problem]  Given $k$ kinds of sensors for different measures and a set $V$ of $n$ locations, each of which is assigned with only one sensor. Assuming that each sensor $e$ has a cost $c(e)>0$ for every $i$-th type.  Given a budget $B>0$, the problem aims to locate these sensors to maximize the information gained with the total cost at most $B$.			
		\end{definition}
		Denote by $R_e^i$ a random variable representing the observation collected from a $i$-type sensor and the information gained  of  a $k$-set $\s$ is 
		\begin{align}
			f(\s)=H(\cup_{e \in supp(\s)} \{R_e^i\})
		\end{align}
		where $H$ is an \textit{entropy function}. 
		\subsection{Results and discussion}
		%%%%%%%%%%%%%%%%%%%%%%%% kIMK%%%%%%%%%%%%%%%%%%%%%%%%%
		\subsubsection{Experiment settings}
		\textbf{For $\kIMK$.} We use the dataset Facebook and set up the model as the recent work~\cite{ksub-stream-icml20}. 
		
		Since the computation of $\sigma(\cdot)$ is \#P-hard~\cite{chen10LT}, we adapt the sampling method in \cite{ksub-stream-icml20,borg} to give an estimation $\hat{\sigma}(\cdot)$ with a $(\lambda,\delta)$-approximation that is:
		\begin{align}
			\Pr[(1+\lambda)\sigma(\s) \geq \hat{\sigma}(\s) \geq (1-\lambda)\sigma(\s)]\geq 1-\delta
			\label{sigmahat}
		\end{align}
		It's said that $\hat{\sigma}(\cdot)$ is $\epsilon$-estimation of $\sigma(\cdot)$ with probability at least $\lambda$. As~\cite{ksub-stream-icml20, Canh_joco21}, in the experiment, we set parameters $\lambda=0.8, \delta=0.2$, $k=3$ and $\epsilon =0.1$ to show a trade-off between solution quality and quantities of queries.
		
		We set $B$ in $\{0.5K, 1K, 1.5K, 2K\}$ to illustrate the expense to influence $k$ topics via social networks is not a small number and set the cost of each element from 1 to 10 according to the Normalized Linear model \cite{Canh_joco21}.    
		
		\textbf{For $\kSPK$.} We use  the dataset Intel Lab~\cite{sensor} to illustrate the $\kSPK$ problem. %Intel Lab sensor data~\cite{sensor} was collected from 2.3 million readings of a Mica2dot weather board sensor system including temperature, humidity, light, and voltage. 
		The data were preprocessed to remove missing fields. Moreover, we set $k=3$, $\epsilon =0.1$ as in the experiment of $\kIMK$, and the cost range from 1 to 10  for the Intel Lab dataset whereas the values of $B$ are fixed at several points from 10 to 50. This setting was due to the number of sensors and the similarity among algorithms. 
		
		\subsubsection{Experiment results} 
		
		To provide a comprehensive experiment, we ran the above algorithms several times and collected results about objective values, the number of queries, and the running time according to the $B$ milestones. For each milestone, the average values were calculated. Figures \ref{fig:f}, and \ref{fig:SS} illustrate the results. 
				
		\textbf{Regarding $\kIMK$}. First, Figure \ref{fig:f}(a) represents the performance of algorithms via values of the objective function $\sigma(\cdot)$. $\RLA$ is equivalent to DS, followed RS, while $\LAA$'s line hits the lowest points. In Figure (a) the gaps between groups $\RLA$-DS, RS, and $\LAA$ seem bigger when $B\ge 1.5K$. 
		
		\begin{figure}[h]
			\centering
			{\includegraphics[width=1\linewidth]{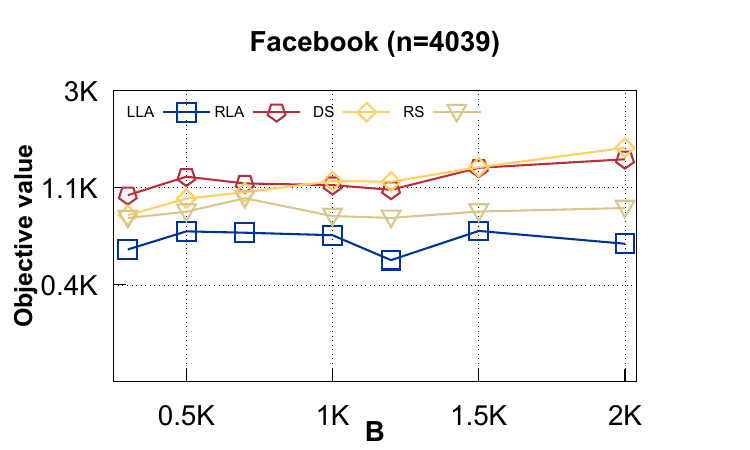}}
			\\
			(a)\\
			{\includegraphics[width=1\linewidth]{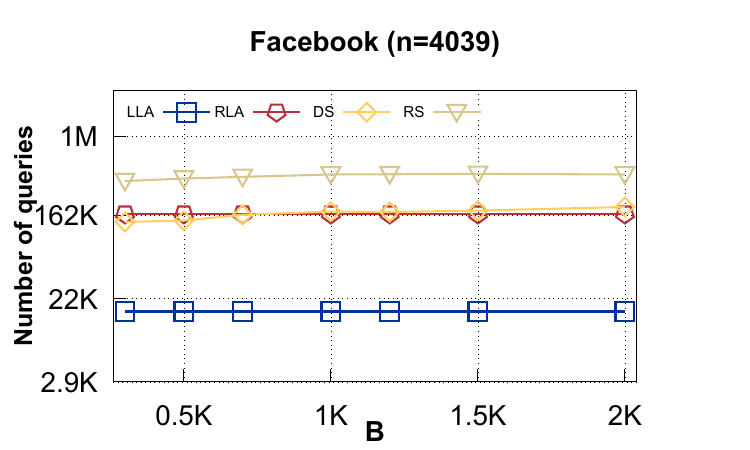}}\\
			(b)
			\\{\includegraphics[width=1\linewidth]{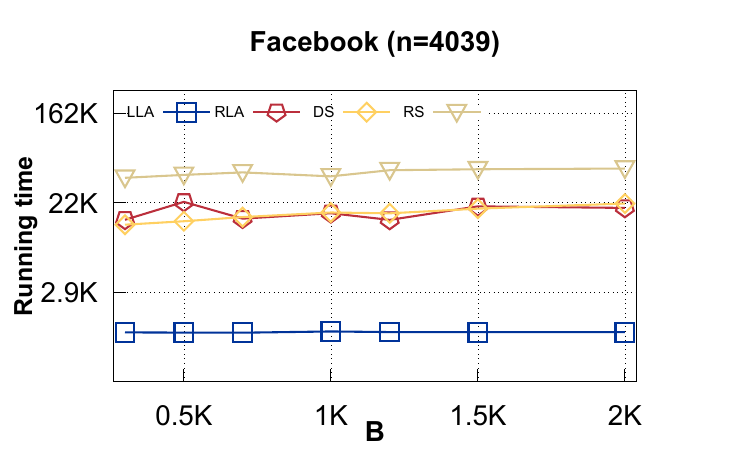}}
			\\ 
			(c)
			\caption{Algorithm results for $\kIMK$ on Facebook: (a) The objective values, (b) The number of queries, (c)~Time consumption}
			\label{fig:f}
		\end{figure}
				
		Second, Figure \ref{fig:f}(b)(c) displays the amounts of queries called and the time needed to run these algorithms. $\LAA$ shows an advantage over others in terms of query complexity. It is sharply from several to dozens of times lower than the remaining. Besides, the number of queries of $\RLA$ is equivalent to DS and lower than RS, respectively. Significantly, these lines explicitly determine and linear over $B$ milestones. Overall, the number of queries of RS is the highest, followed by the group of $\RLA$-DS and $\LAA$, respectively. The experiment indicates the quantities of queries of our algorithms outperform the others.   
		
		As the query complexity directly influences the running time, the representation of the time graph in Figure~\ref{fig:f}(c) looks quite similar to the representation of the query graph in Figure~\ref{fig:f}(b) in which $\LAA$ line was drawn typically lowest. It shows the running time of $\LAA$ is several to dozens of times faster than the others. $\RLA$ runs considerably faster than RS and equivalently to DS.  
		
		The above figures show the trade-off between our proposed algorithms' solution qualities and the query complexities. $\LAA$ tries to target the near-optimal value by dividing the ground into two subsets according to the cost values of elements and reduces query complexity by the filtering condition of the algorithm \ref{alg:1}. Hence, the query complexity is significantly low. Nevertheless, the performance of $\LAA$ regarding solution quality is not high. $\RLA$ enhance $\LAA$ by using $\LAA$ as an input and the decreasing constant threshold. As a result, the objective of $\RLA$ is better than $\LAA$ while the number of queries is higher but still deterministic. $\RLA$ use the threshold $\tau_v$ to upgrade the performance. It leads to the objective value increasing, yet the number of queries also increases. Moreover, when the ground set and $B$ value grow, the solution quality improves while running time and query complexity are linear. This is extremely important when working with big data.
		
		%%%%%%%%%%%%%%%%% kSPK%%%%%%%%%%%%%%%%%%%%%%%%%%%%%%%%%%%%%%
		
		\textbf{Regarding $\kSPK$.}
		\begin{figure}[h]
			\centering
			{\includegraphics[width=1\linewidth]{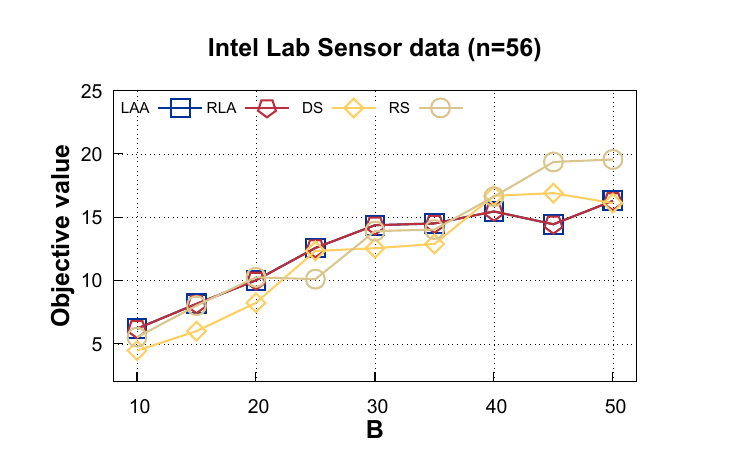}}
			\\(a)\\
			{\includegraphics[width=1\linewidth]{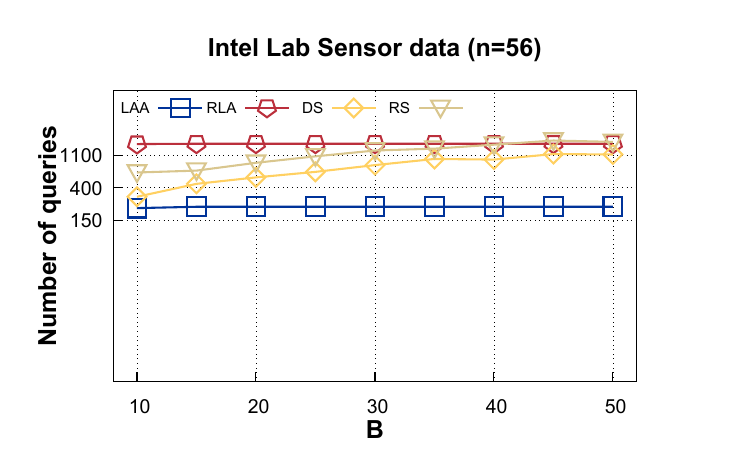}}
			\\(b)\\
			{\includegraphics[width=1\linewidth]{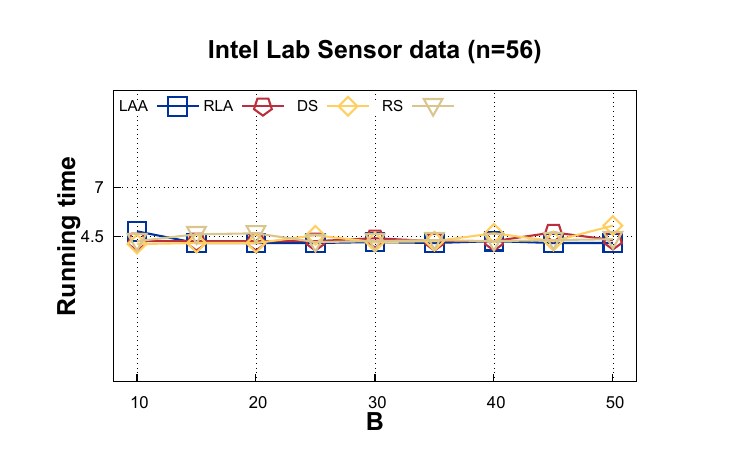}}
			\\(c)\\
			\caption{Algorithm results for $\kSPK$ on Intel Lab: (a) The information gained, (b) the number of queries, (c)~the running time }
			\label{fig:SS}
		\end{figure}
		As can be seen in \ref{fig:SS}(a) the discrimination between objective values of experimented algorithms is not large. $\LAA$ and $\RLA$ seem to overlap while DS and RS fluctuate a bit. When $B$ increases, the gap between these lines becomes larger in which RS states the highest, followed by DS and the group of $\LAA$-$\RLA$, respectively. However, due to the number of nodes being small, on the whole, information gained from these algorithms is almost no different. 
		
		Second, the gap between the number of queries of $\LAA$ and the others in Figure \ref{fig:SS}(b) is significantly large. With the large number of $B$, lines of $\RLA$, DS, and RS tend convergent in which $\RLA$ lies on the remaining, followed by RS and DS, respectively. Regarding Figure~\ref{fig:SS}(c), these lines seem to overlap. Moreover, query lines and timelines of the above algorithms are almost horizontal over $B$'s milestones. This result illustrates the query complexity of our algorithms is linear and equivalent to other ones. 
		
		From two actual uses of $\kIMK$, and $\kSPK$, our solutions are better or equivalent to existing ones while the number of queries reduces, especially when $n$ and $B$ grow. The steady of the proposed linear deterministic ones becomes vital when the data increases. The experiment showed was consistent with the theory. It also indicated the trade-off between our proposed algorithms' solution qualities and the query complexities.  Overall, our proposed algorithms are described to outperform or be comparable to the state-of-the-art.
		\section{Conclusion}
		\label{sec:con}
	This paper works with the problem of maximizing a $k$-submodular function under a knapsack constraint for the non-monotone case. We propose two deterministic algorithms that take just $O(kn)$ query complexity. The core of our algorithms is to keep which elements are over a given appropriate threshold and then choose among them the last elements so that the total cost does not exceed a given budget, $B>0$. 
	
	To investigate the performance of our algorithms in practice, we conducted some experiments on two applications of Influence Maximization and Sensor Placement. Experimental results have shown that our algorithms not only return acceptably reasonable solutions regarding quality requirements but also take a smaller number of queries than state-of-the-art algorithms. However, there are still some open questions, such as how to improve the approximate ratio or the linear query complexity for the $\NkSMK$ problem, that will motivate us in the future.   
		\bibliographystyle{IEEEtran}
		\bibliography{KSE}
		\newpage
		\onecolumn
		\section*{Appendix}	
		\subsection{The proofs of Lemmas and Theorems}\label{sec:proof}
	\begin{proof}[Proof of Lemma~\ref{lem:1}]		
		Due to $f$ might be non-monotone, recap that $(e_j,i_j)$ is the $j$-tuple added into the candidate set $\x$ after the loop of Algorithm \ref{alg:1}. We have 2 sub-cases:
		\begin{itemize}
			\item If $e_j\notin supp(\o_2)$, define an integer number $l\in[k]$ that $l\ne i_j$ and $\o^j_l$ a $k$-set such that $\o^j_l(e)=\o^j_2(e),\forall e\in V_2\setminus\{e_j\}$ and $\o^j_l(e_j)=l$, we have:
			\begin{align}
				f(\o_2^{j-1})-f(\o_2^j)&= f(\o^j_l) -f(\o_2^{j-1}) \nonumber \\
				&- (f(\o_2^j)+ f(\o^j_l)- 2f(\o_2^{j-1}) )
				\\
				& \leq f(\o^j_l) -f(\o_2^{j-1}) \label{ine:1} %\ \ \ (\mbox{due to the pairwise-monotonicity})
				\\
				& \leq f(\x^j_l) -f(\x^{j-1}) \label{ine:2}
				%\ \ \ (\mbox{due to the $k$ submodularity})
				\\
				& \leq f(\x^j) -f(\x^{j-1}) \label{ine:3} 
				%\ \ \ (\mbox{due to the selection of $i_j$ in the algorithm})
			\end{align}
			where the inequality~\eqref{ine:1} is due to the pairwise-monotoncity of $f$, the inequality~\eqref{ine:2} is due to the $k$-submodularity of $f$, and the inequality~\eqref{ine:3} is due to the selection rule of the algorithm.
			The proof is completed. 
			\item If $e_j \in supp(\o_2)$. In this case, if $\o_2^{j-1}(e_j)=i_j$. Due to the  pairwise-monotone property of $f$, there exists $i'\in [k]$ that $f(\x^{j-1} \sqcup (e_j, i')) \geq 0$. Therefore, 
			$$f(\o_2^j)-f(\o_2^{j-1})=0\leq  f(\x^j) -f(\x^{j-1})$$ 
			If  $\o_2^{j-1}(e_j)\neq i_j$, we obtain:
			\begin{align*}
				f(\o^{j-1})-f(\o^j) & = 2f(\o^{j-1}) -2f(\o^{j-1/2}) \\
				&- (f(\o^{j-1})+ f(\o^j)- 2f(\o^{j-1/2}) )
				\\
				& \leq 2f(\o^{j-1}) -2f(\o^{j-1/2})\\
				&\leq 2f(\x^{j}) -2f(\x^{j-1})
			\end{align*}
		\end{itemize}
		The last inequality is due to the $k$-submodularity. Overall, we have $f(\o_2^{j-1})-f(\o_2^j) \leq 2f(\x^{j}) -2f(\x^{j-1})$. Therefore,
		\begin{align*}
			f(\o_2)-f(\o_2^t) &= \sum_{j=1}^t(f(\o^{j-1})-f(\o^{j}))  
			\\
			& \leq 2\sum_{j=1}^t(f(\x^{j})-f(\x^{j-1}))  \leq 2f(\x^t) 
		\end{align*}
		
		The proof is completed. 
	\end{proof}
\begin{proof}[Proof of Lemma~\ref{lem:2}]
	Recap $\x=\x^t$. If $c(\x^t) \leq B$, $\x'=\x^t$ and the Lemma holds. Therefore, we must consider the case $c(\x^t)>B$. We get:
	\begin{align}
		f(\x^t) - f(\x^Q) &=\sum_{j=Q+1}^T\Delta_{(e_j,i_j)}f(\x^{j-1})
		\\
		&\geq \sum_{j=Q+1}^T c(e_j)\frac{f(\x^{j-1})}{B} \label{ine:4}
		\\ 
		&\geq \sum_{j=Q+1}^Tc(e_j)\frac{f(\x^Q) }{B}  \label{ine:4-1}
		\\
		&\geq c(\x') \frac{f(\x^Q)}{B}
	\end{align}
	where the inequality~\eqref{ine:4} is due to the selection rule of a tuple $(e, i_j)$ into $\x$ according to the condition  at Line~\ref{cond2} of Algorithm~\ref{alg:1}. The inequality~\eqref{ine:4-1} is due to $f(\cdot)\ge 0$. 
	
	Since $\x'$ is chosen from $\x$ so that its total cost is closest to $B$ and each element $e \in supp(\x)$ has the cost at most $B/2$, thus: $$c(\x') > B- \frac{B}{2}\geq \frac{B}{2}$$
	It implies that $f(\x^t)-f(\x^Q) \geq f(\x^Q)/2$. Hence $f(\x^Q)\leq 2f(\x^t)/3$. In the other hand, due to the $k$-submodularity of $f$ we have $f(\x^t)\leq f(\x^Q)+ f(\x')$. Thus,
	\begin{align}
		f(\x')\geq f(\x^t)-f(\x^Q)\geq \frac{f(\x^t)}{3} 
	\end{align}
	The proof is completed.
\end{proof}
\begin{proof}[Proof of Lemma~\ref{lem:3}]
We have:
\begin{align}
	f(\o_2)-f(\x^t) &= f(\o_2)-f(\o_2^t)+f(\o_2^t)-f(\x^t)\\
	&\le 2f(\x^t) +f(\o_2^t)-f(\x^t) \ \  \ (\mbox{Due to Lemma~\ref{lem:1}})\label{ine:5}\\
	&\le 2f(\x^t)+\sum_{e\in supp(\o_2^t)\setminus supp(\x^t)}\Delta_{(e,i_j),\forall i\in[k]}f(\x^t)\label{ine:6}\\ 
	&\le 2f(\x^t)+ \sum_{e\in supp(\o_2^t)\setminus supp(\x^t)}\frac{c(e)f(\x^t)}{B}\label{ine:7}\\
	&\le 2f(\x^t)+ \frac{Bf(\x^t)}{B}=3f(\x^t)\label{ine:8}
\end{align}
%\begin{align}
%	f(\o_2^t)-f(\s^t) &= f(\u^t \sqcup \s^t)- f(\s^t)
%	= \sum_{i=1}^{r}(f(\u^t_i)-f(\u^t_{i-1}))  \nonumber
%	\\
%	&\leq  \sum_{i=1}^{r}(f(\s^{<{u_i}} \sqcup (u_i, j_i))-f(\s^{<{u_i}})  \label{ine:5}
%		\\
%	&\leq \sum_{i=1}^{r} \Delta_{(u_i, j_i)}f(\s^{<{u_i}}) \label{ine:6}
%		\\
%		& \leq \sum_{i=1}^{r}c(u_i)\frac{f(\s^{<{u_i}})}{B} \label{ine:7}
%		\\
%		& \leq \sum_{i=1}^{r}c(u_i)\frac{f(\s^t)}{B} \label{ine:8}
%		\\
%		&\leq c(\u^t)f \frac{(\s^t)}{B} \leq f(\s^t) 
%		\end{align}
where the inequality%~\eqref{ine:5} and
~\eqref{ine:6} is due to the $k$-submodularity of $f$, %the inequality \eqref{ine:6} due to the definition of $\s^{<{u_i}}$, 
the inequality~\eqref{ine:7} is due to the selection of the algorithm. Thus, we have: $f(\o^t_2) \leq 4f(\x^t)$ or $f(\o^t_2) \leq 4f(\x)$. The proof is completed. 
\end{proof}
\begin{proof}[Proof of Lemma~\ref{lem:4}]
	Applying the Lemmas \ref{lem:1} and \ref{lem:3}, with $j=t$, we have:
	\begin{align*}
		f(\o_2)-f(\x^t) &= f(\o_2) - f(\o_2^t) + f(\o_2^t) - f(\x^t)
		\\
		& \leq 2f(\x^t) + f(\o_2^t) - f(\x^t) 
		\\
		& \leq 5f(\x^t)
	\end{align*}
	Thus $f(\x^t) \geq f(\o_2)/6$. Combine the above implication with the Lemma \ref{lem:2}, we have $f(\x')\geq f(\x^t)/3\geq f(\o_2)/18$.
\end{proof}
\begin{proof}[Proof of Theorem~\ref{theo:alg1}]
The algorithm scans only once over the ground set, and each element $e$ has $k$ queries to find the position $i_e$. Therefore the number of  queries is $nk$.
We now prove the approximation ratio of the algorithm. By the selection of $(e_{m}, i_m)$ and the $\o_1$ contains at most one element so $f(\o_1)\leq f((e_{m}, i_m))$. 
By the definition of $\o_1', \o_2'$ and the $k$-submodularity of $f$, we obtain:
\begin{align}
	f(\o) & \leq   f(\o'_1)+ f(\o'_2) \label{ie-theo1-1}
	\\
	& \leq f(\o_1)+ f(\o_2) \label{ie-theo1-2}
	\\
	& \leq f((e_{m}, i_m))+18f(\x') \leq 19 f(\s) \label{ie-theo1-3}
\end{align}
The proof was completed.  
\end{proof}
\begin{proof}[Proof of Lemma~\ref{lem:alg2.1}] Due to the same selection rule between $(e, i_v)$ of Algorithm~\ref{alg:2} and $(e, i_e)$ of Algorithm~\ref{alg:1}, we have the same result with Lemma~\ref{lem:1}, i.e.,  $f(\o)-f(\o^q)\leq 2f(\s_v)$. Thus:
	\begin{align} 
		f(\o)-f(\s_v) & =f(\o)-f(\o^q) + f(\o^q)-f(\s_v)
		\\
		& \leq  2f(\s_v)  +  \sum_{j=1}^{r}(f(\u_j)-f(\u_{j-1}))  %\eqref{ine:t3-0}
		\\
		& \leq 2f(\s_v)  + \sum_{j=1}^{r}\Delta_{(u_j, i_j)} f(\s_v^{<{u_j}}) \label{ine:t3-1}
		\\
		& \leq 2f(\s_v)  +\sum_{j=1}^{r}c(u_j)\tau_v \label{ine:t3-2}
		\\
		& \leq 2f(\s_v)  +c(\o)\tau_v
		\label{ine:t3-3}
	\end{align}
	where the inequality~\eqref{ine:t3-1} is due to the $k$-submodularity, the  inequality~\eqref{ine:t3-2} is due to the definition of $\s_v^{<o}$, and the inequality \eqref{ine:t3-3} is due to the definition of $\u$ and $\o$. Thus, the proof is completed. 

\end{proof}
\begin{proof}[Proof of Theorem~\ref{theo:alg2}]
		The algorithm needs  $nk$ queries to call $\LAA$ and uses only 1-pass over the ground set for finishing the outer loop (Line 3-11). For each incoming element, it takes at most $k\cdot	\lceil \log_{(1+\epsilon)}(19) \rceil $ queries for updating $\s_v, v\in A$. Combine all tasks, the required number of queries at most: %$$nk+nk \log_{(1+\epsilon)}(19)\leq nk+nk \frac{\log(19)}{\epsilon}=O(\frac{nk}{\epsilon}).$$
	\begin{align*}
		nk+nk	\lceil \log_{(1+\epsilon)}(19) \rceil & \leq nk+nk (1+ \log_{(1+\epsilon)}(19))\\
		&=2nk+ nk \frac{\ln (19)}{\ln(1+\epsilon)} 
		\\
		&\leq 2nk+ nk \frac{\ln (19)}{\ln\frac{1}{1-\frac{\epsilon}{2}}}\\
		&= 2nk- nk \frac{\ln (19)}{\ln(1-\frac{\epsilon}{2})}\\
		&\leq  2nk+\frac{2}{\epsilon} nk\ln(19)=O(\frac{nk}{\epsilon}) 
	\end{align*}
	We now show the approximation ratio of the algorithm.	By Theorem \ref{theo:alg1}, we have $ \Gamma \leq \opt \leq 19. \Gamma$. Therefore, there exists an integer number $v\in A$ so that $ \opt/(1+\epsilon) \leq v \leq \opt $. 
	We have:
	\begin{align}
		f(\s^j_v)=\sum_{i=1}^j (f(\s^{i}_v)-f(\s^{i-1}_v)) \geq \sum_{i=1}^j c(e_i)\theta_v=c(\s^j_v)\theta_v \label{ine:total}
	\end{align}
	%	We denote by $\s_v^{<e}$ as $\s_v$ immediately before $e$ is processed. 
	We consider the following cases:
	\\
	\textbf{Case 1.} There exists an element $o \in supp(\o)\setminus supp(\s_v)$ so that $\Delta_{(o,\o(o))}f(\s_v^{<o})\geq \tau_v$ and $c(\s_v^{<o})+c(o)>B$. Recall $(e_m, i_m)=\arg\max_{e\in V, i\in [k]}f((e, i))$, we have:
	\begin{align}
		f(\s_{final})&\geq \max\{f(\s_v), f((e_{m}, i_{m}))\} 
		\\
		&\geq \max\{f(\s_v^{<o}), f((o, \o(o)))\} 
		\\
		&
		\geq \frac{f(\s_v^{<o}) +f((o, \o(o)))}{2}
		\\
		& \geq \frac{f(\s_v^{<o} \sqcup (o, \o(o))}{2}
		\\
		&=\frac{\Delta_{(o,\o(o))}f(\s_v^{<o}) +f(\s^{<o}_v)}{2}
		\\
		&\geq \frac{\tau_v c(o) +\tau_v c(\s^o_v)}{2}\geq \frac{B\tau_v}{2} = \frac{v}{5}\\
		&\geq\frac{\opt}{5(1+\epsilon)}
		\\
		& \geq (\frac{1}{5}-\epsilon)\opt
	\end{align}
	
	\textbf{Case 2.} There is no such an element $o$ like Case 1.
	By the  Lemma \ref{lem:alg2.1}, $c(\o)\leq B$, $\tau_v=2v/(5B)$, and $f(\o)=\opt \geq v$ we have:
	\begin{align}
		f(\o)&\leq 3f(\s_v)  +B\tau_v\leq 3f(\s_v) +2B v/(5B) \\ 
		&\leq 3f(\s_v) + 2\opt/5
	\end{align}
	It implies: $\opt\leq3f(\s_v) + 2\opt/5$, thus $ \opt\leq 5f(\s_v)$. Finally,  $f(\s_{final})\geq f(\s_v) \geq \opt/5$. By combining the two above cases, we obtain the proof.
\end{proof}
	\end{document}